\documentclass[%
 reprint,
 superscriptaddress,
 amsmath,amssymb,
 aps,
 pra,
]{revtex4-2}

\usepackage[utf8]{inputenc} 
\usepackage{amsmath,amsfonts,amssymb,mathrsfs,amsthm,cases} 
\usepackage{bm, bbm, dsfont} 
\usepackage{physics} 
\usepackage{graphicx,subfig} 
\usepackage{diagbox,threeparttable,dashrule,booktabs,dcolumn} 
\usepackage{xcolor} 
\usepackage{url} 
\usepackage{hyperref} 
\hypersetup{colorlinks,linkcolor={blue},citecolor={blue},urlcolor={red}}
\usepackage{cleveref} 

\crefname{equation}{Eq.}{Eqs.}
\crefname{section}{Sec.}{Secs.}
\crefname{definition}{Definition}{Definitions}
\crefname{proposition}{Proposition}{Propositions}
\crefname{lemma}{Lemma}{Lemmas}
\crefname{theorem}{Theorem}{Theorems}
\crefname{corollary}{Corollary}{Corollaries}
\crefname{conjecture}{Conjecture}{Conjectures}
\crefname{claim}{}{Claims}
\crefname{example}{Example}{Examples}

\newtheorem{definition}{Definition}
\newtheorem{proposition}{Proposition}
\newtheorem{lemma}{Lemma}

\newtheorem{theorem}{Theorem}
\newtheorem{corollary}{Corollary}

\newtheorem{example}{Example}

\newtheorem{manualtheoreminner}{Corollary}
\newenvironment{manualtheorem}[1]{%
  \renewcommand\themanualtheoreminner{#1}%
  \manualtheoreminner
}{\endmanualtheoreminner}

\newcommand{\proj}[1]{|{#1}\rangle \langle {#1}|}

\newcommand{\nc}{\newcommand}

\def\bpf{\begin{proof}}
\def\epf{\end{proof}}
\def\bea{\begin{eqnarray}}
\def\eea{\end{eqnarray}}
\def\beq{\begin{equation}}
\def\eeq{\end{equation}}
\def\bal{\begin{aligned}}
\def\eal{\end{aligned}}
\def\bma{\begin{pmatrix}}
\def\ema{\end{pmatrix}}

\def\bigox{\bigotimes}
\def\dg{\dagger}

\def\ox{\otimes}

\def\a{\alpha}
\def\b{\beta}

\nc{\bbA}{{\mathbb A}}  \nc{\bbB}{{\mathbb B}}  \nc{\bbC}{{\mathbb C}}
\nc{\bbD}{{\mathbb D}}  \nc{\bbE}{{\mathbb E}}  \nc{\bbF}{{\mathbb F}}
\nc{\bbG}{{\mathbb G}}  \nc{\bbH}{{\mathbb H}}  \nc{\bbI}{{\mathbb I}}
\nc{\bbJ}{{\mathbb J}}  \nc{\bbK}{{\mathbb K}}  \nc{\bbL}{{\mathbb L}}
\nc{\bbM}{{\mathbb M}}  \nc{\bbN}{{\mathbb N}}  \nc{\bbO}{{\mathbb O}}
\nc{\bbP}{{\mathbb P}}  \nc{\bbQ}{{\mathbb Q}}  \nc{\bbR}{{\mathbb R}}
\nc{\bbS}{{\mathbb S}}  \nc{\bbT}{{\mathbb T}}  \nc{\bbU}{{\mathbb U}}
\nc{\bbV}{{\mathbb V}}  \nc{\bbW}{{\mathbb W}}  \nc{\bbX}{{\mathbb X}}
\nc{\bbY}{{\mathbb Y}}  \nc{\bbZ}{{\mathbb Z}}  

\nc{\bA}{{\bf A}}  \nc{\bB}{{\bf B}}  \nc{\bC}{{\bf C}}
\nc{\bD}{{\bf D}}  \nc{\bE}{{\bf E}}  \nc{\bF}{{\bf F}}
\nc{\bG}{{\bf G}}  \nc{\bH}{{\bf H}}  \nc{\bI}{{\bf I}}
\nc{\bJ}{{\bf J}}  \nc{\bK}{{\bf K}}  \nc{\bL}{{\bf L}}
\nc{\bM}{{\bf M}}  \nc{\bN}{{\bf N}}  \nc{\bO}{{\bf O}}
\nc{\bP}{{\bf P}}  \nc{\bQ}{{\bf Q}}  \nc{\bR}{{\bf R}}
\nc{\bS}{{\bf S}}  \nc{\bT}{{\bf T}}  \nc{\bU}{{\bf U}}
\nc{\bV}{{\bf V}}  \nc{\bW}{{\bf W}}  \nc{\bX}{{\bf X}}
\nc{\bY}{{\bf Y}}  \nc{\bZ}{{\bf Z}}  

\nc{\cA}{{\cal A}}  \nc{\cB}{{\cal B}}  \nc{\cC}{{\cal C}}
\nc{\cD}{{\cal D}}  \nc{\cE}{{\cal E}}  \nc{\cF}{{\cal F}}
\nc{\cG}{{\cal G}}  \nc{\cH}{{\cal H}}  \nc{\cI}{{\cal I}}
\nc{\cJ}{{\cal J}}  \nc{\cK}{{\cal K}}  \nc{\cL}{{\cal L}}
\nc{\cM}{{\cal M}}  \nc{\cN}{{\cal N}}  \nc{\cO}{{\cal O}}
\nc{\cP}{{\cal P}}  \nc{\cQ}{{\cal Q}}  \nc{\cR}{{\cal R}}
\nc{\cS}{{\cal S}}  \nc{\cT}{{\cal T}}  \nc{\cU}{{\cal U}}
\nc{\cV}{{\cal V}}  \nc{\cW}{{\cal W}}  \nc{\cX}{{\cal X}}
\nc{\cY}{{\cal Y}}  \nc{\cZ}{{\cal Z}}  

\def\ox{\otimes}
\def\bigox{\bigotimes}
\def\dg{\dagger}

\begin{document}


\title{The additivity of states uniquely determined by marginals} 

\author{Yi Shen}
\email[]{yishen@jiangnan.edu.cn}
\affiliation{School of Science, Jiangnan University, Wuxi Jiangsu 214122, China}

\author{Lin Chen}
\email[Corresponding author: ]{linchen@buaa.edu.cn}
\affiliation{School of Mathematical Sciences, Beihang University, Beijing 100191, China}
\affiliation{International Research Institute for Multidisciplinary Science, Beihang University, Beijing 100191, China}

\date{\today} 

\begin{abstract}
The pure states that can be uniquely determined among all (UDA) states by their marginals are essential to efficient quantum state tomography. We generalize the UDA states from the context of pure states to that of arbitrary (no matter pure or mixed) states, motivated by the efficient state tomography of low-rank states. 
We call the \emph{additivity} of $k$-UDA states for three different composite ways of tensor product, if the composite state of two $k$-UDA states is still uniquely determined by the $k$-partite marginals for the corresponding type of tensor product. 
We show that the additivity holds if one of the two initial states is pure, and present the conditions under which the additivity holds for two mixed UDA states. One of the three composite ways of tensor product is also adopted to construct genuinely multipartite entangled (GME) states. Therefore, it is effective to construct multipartite $k$-UDA state with genuine entanglement by uniting the additivity of $k$-UDA states and the construction of GME states.
\end{abstract}


\maketitle


\section{Introduction}
\label{sec:intro}

In quantum mechanics the correlation between the whole and its parts reflects a significant difference from its classical counterpart \cite{Wyderka_phd}. A remarkable example is the Bell state, each of whose single-party reduced state is maximally mixed. It means that no useful information of each party can be obtained by local measurements. For a global state prepared, one can obsesrve its parts by implementing measurements. More importantly, can we infer or even reconstruct the global state with given measurement results? This task is known as the quantum state tomography (QST) \cite{QST2012}. Performing QST 
based on the reduced states is a common and efficient method to characterize the global state \cite{QSTviaRDM2017,pxdxduda,Xin2019}. It closely connects to another essentail topic in quantum information, namely the marginal problem, which stems from quantum chemistry \cite{FDM1963}. The marginal problem asks whether there is a global state compatible with a given set of multipartite marginal reductions \cite{QMA_Dardo_2022}. If the answer is positive, it is interesting to further study whether such a global state is uniquely determined. This uniqueness issue plays a core role in the efficient QST \cite{Cramer2010,QSTviaRDM2017}, because the general case of QST requires a large number of observables as the dimension of the quantum system increases \cite{QST_CMP}.

If the set of marginal reductions is generated from a given pure state, the uniquess issue is specified as whether there is another global state sharing all the same reductions as the given state \cite{3qbuda}. According to the scope of the discussion, this problem is divided into two cases, namely uniquely determined among pure states (UDP) \cite{2of3udp2004} and uniquely determined among all states (UDA) \cite{udaubound}. The former means there is no other pure state satisfying the desired requirement, and the latter means there is no other state satisfying the desired requirement. Such uniquely determined states are of practical interest for the following reasons. First, they make tomography meaningful and efficient as we mentioned above. Second, they are closely related to the unique ground (UG) state of a Hamiltonian that may be obtained by engineering this Hamiltonian and then cooling down the system \cite{UGS_Huber,udp_ugs}. Third, the relation of the three classes of UDP, UDA and UG states respectively implies a hierarchy of topological order \cite[Fig. 4]{tpo2018}. Topological states are widely used to construct topological stabilizer codes \cite{Liao2021,Liao2022}, for example the toric codes and the surface codes. These codes are typically defined in the systems of large number of parties. Hence, it is necessary to study the UDP and UDA states of large number of parties. This is one of our motivations for proposing the concept of \emph{additivity} which allows to generate UDA states of larger number of parties.


Bipartite states generally cannot be fixed by their single-party reductions. Hence, the first non-trivial case of UDA states should be tripartite states. It was first shown that almost every pure three-qubit state is completely determined by its bipartite reduced states \cite{3qbuda}. Subsequently, it was shown that almost every pure multipartite state (whose local dimensions are all equal) is UDA by its reductions of a fraction of the parties, and the fraction was specified as less than about two-thirds of the parties \cite{udaubound}. This fraction was further improved to be just over half the parties \cite{npudalbound}. With the prior knowledge that the given state is pure, a stronger result comes out that almost every tripartite pure state is uniquely determined by only two of three bipartite marginals \cite{2of3udp2004}. Similar conclusions were extended to  the UDA states in some tripartite systems \cite{udaobservable,pxdxduda}. These results imply that only a part of reductions is sufficient to ensure the uniqueness. Thus, it is interesting to find which reductions are enough to ensure the uniqueness \cite{4p2udp}.
Note here that the term ``almost every state'' is adopted to characterize the set of UDP or UDA states, which means excluding a measure zero set from the overall set being considered. In some cases such measure zero set is explicit \cite{GGHZ2008,4p2udp,GGHZ2009}. 
For example, only the $n$-qubit generalized GHZ states and their local unitary equivalents cannot be fixed by the $(n-1)$-qubit reductions \cite{GGHZ2009}.


Both UDP and UDA states are typically specified as pure states. Nevertheless, pure states are extremely unstable, which leads to mixed states being more common in the laboratory. It motivates us to consider the uniqueness issue in terms of arbitrary (pure or mixed) states. We define the mixed UDA states in Definition \ref{def:guda} by generalizing the typical definition. One can verify whether an arbitrary state is UDA as the flow chart depicted in Fig. \ref{fig:def-uda}. The QST of low-rank states has been studied recently \cite{QSTviaCS,QST_rank2016,Yuen2023}. Analogous to pure UDA states, mixed UDA states can also make the QST of low-rank states more efficient. 
In this paper we study the generation of UDA states in the systems of large number of parties and high local dimensions. Specifically, we consider three different composite ways of tensor product defined by Definition \ref{def:prods}. Each type of products corresponds to a way of system composition. We illustrate such system compositions in Fig. \ref{fig:tensors}. The first two products are known as the tensor product and the Kronecker product respectively, which can expand the number of parties and enlarge local dimensions respectively from Figs. \ref{fig:tensor1} and \ref{fig:tensor2}. The third one combines the first two types of products by applying the Kronecker product on a part of subsystems, which can expand the number of parties and enlarge local dimensions simultaneously from Fig. \ref{fig:tensor3}. To characterize mixed UDA states, we extend two essential properties of pure UDA states to the case of mixed ones in Lemmas \ref{le:luequiv} - \ref{le:ktok+1}.
Then we call the \emph{additivity} of UDA states under each type of products if the corresponding product of two UDA states is still UDA by the marginals of the same number of parties. 
By virtue of the specific expressions formulated in Lemma \ref{le:preduced}, we show that the UDA states admit the additivity if one of the initial two states is pure, in Theorem \ref{thm:add-p+m}. Specially, the additivity of UDA states holds in the context of pure states. If the initial two states are both mixed UDA states, we propose the conditions for the additivity in Lemma \ref{le:prodadd}, and conjecture there is generally no additivity in this setting by Lemma \ref{le:ptzero}. 
Furthermore, uniting the construction of GME states proposed in Ref. \cite{geshen2020} and the additivity of UDA states, we  derive an operational approach to construct multipartite UDA states with genuine entanglement in Proposition \ref{prop:geuda}. We illustrate the construction process by Fig. \ref{fig:geuda}. By repeating this process, we construct a class of GME states which are uniquely determined by two-particle correlations only, in Example \ref{ex:geuda-1}.


The remainder of this paper is organized as follows. In Sec. \ref{sec:pre} we clarify some notations and present necessary definitions. In Sec. \ref{sec:prop} we extend two essential properties of pure $k$-UDA states to mixed $k$-UDA states. In Sec. \ref{sec:add} we propose the concept of additivity of $k$-UDA states, and show that $k$-UDA states admit the additivity in the case when one of the two initial states is pure. We further discuss the additivity of two mixed $k$-UDA states in this section.
In Sec. \ref{sec:geuda}, we provide an effective method to construct genuinely entangled UDA states.
Finally, the concluding remarks are given in Sec. \ref{sec:con}.

\section{Preliminaries}
\label{sec:pre}

In this section we shall clarify some notations for convenience, and formulate necessary definitions on mixed UDA states and different types of tensor products.

First, we introduce some notations for clear expression. For any positive integer $m$, denote by $[m]$ the set as $\{1,2,\cdots,m\}$. Let $\cS$ be a subset of $[m]$. Then denote by $\cS^c$ the complement of $\cS$ in $[m]$, i.e. $[m]\backslash \cS$.
Suppose that $A_1,\cdots,A_m$ are $m$ systems associated with the Hilbert spaces $\cH_{A_1},\cdots,\cH_{A_m}$ respectively. For any subset $\cS\subseteq [m]$, we denote the composite system $\bigox_{i\in\cS}A_i$ as $A_{\cS}$ associated with the Hilbert space $\bigox_{i\in\cS}\cH_{A_i}$. Next, suppose $B_1,\cdots,B_n$ are another $n$ systems associated with the Hilbert spaces $\cH_{B_1},\cdots,\cH_{B_n}$ respectively. Let $l=\min\{m,n\}$ and $\cS$ be a subset of $[l]$.
Analogously, we denote the composite system $\bigox_{i\in\cS}(A_i\ox B_i)$ as $(AB)_{\cS}$ associated with the Hilbert space $\bigox_{i\in\cS}(\cH_{A_i}\ox \cH_{B_i})$.
For more simplicity, we denote the composite system $\bigox_{i\in[\ell]}(A_i\ox B_i)$ as $(C_1,\cdots,C_\ell)$, where $\ell=\max\{m,n\}$, $C_j=(A_jB_j)$ for each $1\leq j\leq \min\{m,n\}$, and $C_j$ is $A_j$ or $B_j$ for each $\min\{m,n\}<j\leq \ell$.


Second, we define the states that can be uniquely determined by their $k$-partite marginal reductions, in terms of arbitrary states rather than pure states only. Since it is generalized from the definitions for pure UDP states and pure UDA states, we also present the typical definitions on pure states as follows.

\begin{definition}
\label{def:udp-uda}
(i) For pure state $\ket{\psi}$, if there is no pure state $\ket{\phi}(\neq\ket{\psi})$ having all the same $k$-partite marginals as $\ket{\psi}$, then $\ket{\psi}$ is called $k$-uniquely determined among pure ($k$-UDP) states. 
    
(ii) For pure state $\rho\equiv\proj{\psi}$, if there is no (pure or mixed) state $\sigma(\neq \rho)$ having all the same $k$-partite marginals as $\rho$, then $\rho$ is called $k$-uniquely determined among all ($k$-UDA) states.
\end{definition}

By definition it is direct to conclude that pure state $\ket{\psi}$ must be $k$-UDP if it is $k$-UDA. Nevertheless, the converse is not obvious. In Ref. \cite{QSTviaRDM2017} the authors firstly constructed a four-qubit pure state which is $2$-UDP but not $2$-UDA. This example reveals that the set of $k$-UDA states is strictly included in the set of $k$-UDP states, see Fig. \ref{fig:uda-udp}. 
Due to the existence of $k$-UDP but not $k$-UDA states, the UDA property shows more essential uniqueness which could play more valuable role in quantum information processing tasks, for example the QST without prior knowledge.   

\begin{figure}[htbp]
\centering
\includegraphics[width=0.48\textwidth]{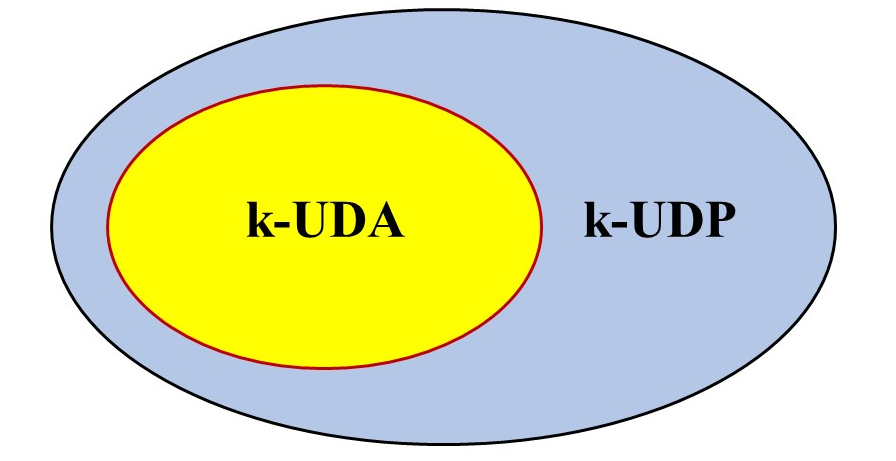}
\caption{The set of pure $k$-UDA states is strictly included in that of pure $k$-UDP states. An example that is $k$-UDP but not $k$-UDA has been proposed in Ref. \cite{QSTviaRDM2017}.}
\label{fig:uda-udp}
\end{figure}

According to Definition \ref{def:udp-uda}, the $k$-UDP states require themselves to be pure states, while the $k$-UDA states can be extended to the case of mixed states by supposing the initial state $\rho$ to be a mixed one.
Then we present the general definition for $k$-UDA states as below. 

\begin{definition}
\label{def:guda}
For an arbitrary state $\rho$ (no matter pure or mixed), if there is no other state $\sigma$ having all the same $k$-partite marginals as $\rho$, then $\rho$ is called $k$-uniquely determined among all ($k$-UDA) states.
\end{definition}

Note that, unless stated otherwise, the $k$-UDA states in this paper is arbitrary (no matter pure or mixed).
By Definition \ref{def:guda} one can verify whether a state is $k$-UDA as the flow chart depicted in Fig. \ref{fig:def-uda}.

\begin{figure}[htbp]
\centering
\includegraphics[width=0.48\textwidth]{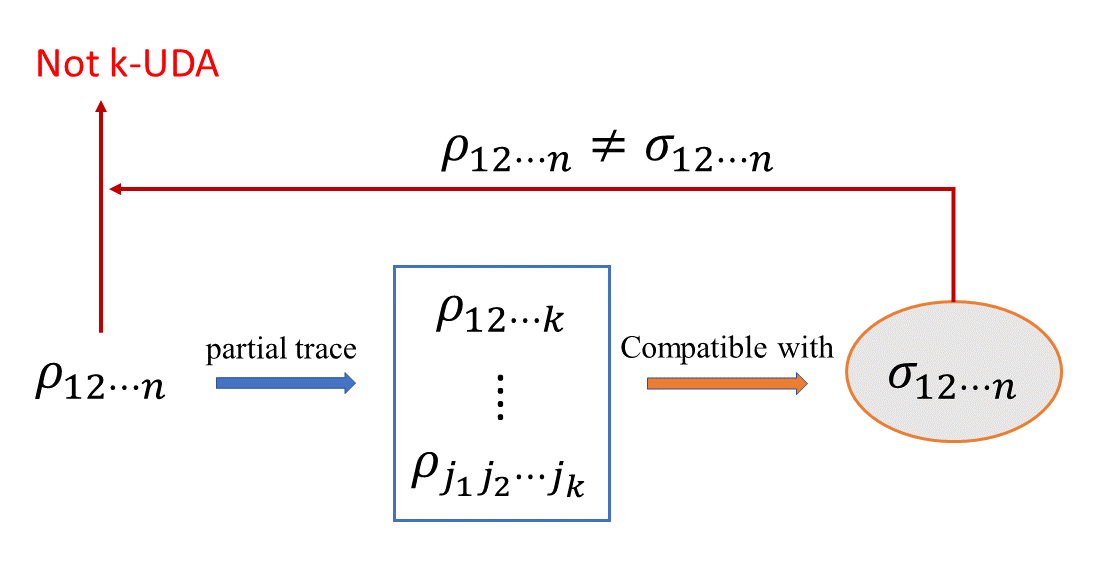}
\caption{The oval represents the compatible set of those $k$-partite marginals from $\rho_{1,\cdots,n}$. $\rho_{1,\cdots,n}$ is not $k$-UDA, if there exists a state $\sigma_{1,\cdots,n}(\neq \rho_{1,\cdots,n})$ in the oval.}
\label{fig:def-uda}
\end{figure}

It follows from Ref. \cite{npudalbound} that the marginals of fewer than half of the parties are not sufficient for the uniqueness among all states. 
It implies that some generic states of fewer parties may transition from a $k$-UDA state to a not $k$-UDA one as the number of parties increases, for a fixed integer $k$. 
Thus, it is necessary to construct $k$-UDA states in a system of more parties and higher local dimensions. In view of this, we introduce several constructions of multipartite states based on different composite ways of tensor product. They were originally proposed to construct GME states, and have been shown to be effective in constructing GME states \cite{geshen2020}. These constructions are operational because the tensor product of two states can be physically realized. We formulate the definitions of three composite ways of tensor product as follows.

\begin{definition}
    \label{def:prods}
Suppose that $\cH_{A_1}\ox \cdots\ox \cH_{A_m}$ and $\cH_{B_1}\ox \cdots\ox \cH_{B_n}$ are two multipartite Hilbert spaces, and assume $m\leq n$ without loss of generality. Let $\rho$ be an $m$-partite state supported on the former Hilbert space and $\sigma$ be an $n$-partite state supported on the latter one. According to different composite ways, there are generally three types of tensor products of $\rho$ and $\sigma$.

(i) The first composite system is defined as
\begin{equation}
    \label{eq:tproddef}
    \cH_{A_1}\ox\cdots\ox\cH_{A_m}\ox\cH_{B_1}\ox\cdots\ox\cH_{B_n}
\end{equation}
which is physically regarded as an $(m+n)$-partite system. The composite state denoted by $\rho\ox\sigma$ supported on the Hilbert space in Eq. \eqref{eq:tproddef} is typically referred to as the tensor product of $\rho$ and $\sigma$. 

(ii) The second composite system is defined as 
\begin{equation}
    \label{eq:kproddef} 
    \bal
    &~~\quad\bigox_{j=1}^m \big(\cH_{A_j}\ox\cH_{B_j}\big) \bigox_{j=m+1}^n \cH_{B_j} \\
    &:= \cH_{(AB)_1}\ox\cdots\ox \cH_{(AB)_m}\ox\cH_{B_{m+1}}\ox\cdots\ox\cH_{B_n}
    \eal
\end{equation}
which is physically regarded as an $n$-partite Hilbert space. The composite state denoted by $\rho\ox_K\sigma$ supported on the Hilbert space in Eq. \eqref{eq:kproddef} is referred to as the Kronecker product of $\rho$ and $\sigma$. 

(iii) The third composite system is defined as $\forall~l<m$,
\begin{equation}
    \label{eq:defkcprof}
    \bal
    &~~\quad\bigox_{j=1}^l \big(\cH_{A_j}\ox\cH_{B_j}\big)\bigox_{j=l+1}^m \cH_{A_j} \bigox_{j=l+1}^n \cH_{B_j} \\
    &:= \cH_{(AB)_1}\ox\cdots\ox \cH_{(AB)_l}\ox\\
    &~~\quad\cH_{A_{l+1}}\ox\cdots\ox\cH_{A_m}\ox\cH_{B_{l+1}}\ox\cdots\ox\cH_{B_n}
    \eal
\end{equation}
which is physically regarded as an $(m+n-l)$-partite Hilbert space. Assume the systems $(AB)_1,\cdots,(AB)_l$ as $C_1,\cdots,C_l$ for simplicity. Then we call the composite state denoted by $\rho\ox_{K_c}\sigma$ supported on the Hilbert space in Eq. \eqref{eq:defkcprof} as the $K_c$-product of $\rho$ and $\sigma$. 
\end{definition}

To better understand the three different composite ways of tensor product, we illustrate Definition \ref{def:prods} by Fig. \ref{fig:tensors}. Each subfigure in Fig. \ref{fig:tensors} corresponds to one composite way of tensor product defined by Definition \ref{def:prods}. Particularly, by Fig. \ref{fig:tensor3} we obsesrve that the proposed $K_c$-product can expand the number of parties and enlarge local dimensions simultaneously.

\begin{figure}[htbp]
    \centering
    \subfloat[The typical tensor product corresponding to Def. \ref{def:prods} (i). It generates an $(m+n)$-partite system.]
        {
        \label{fig:tensor1}
        \includegraphics[width=0.45\textwidth]{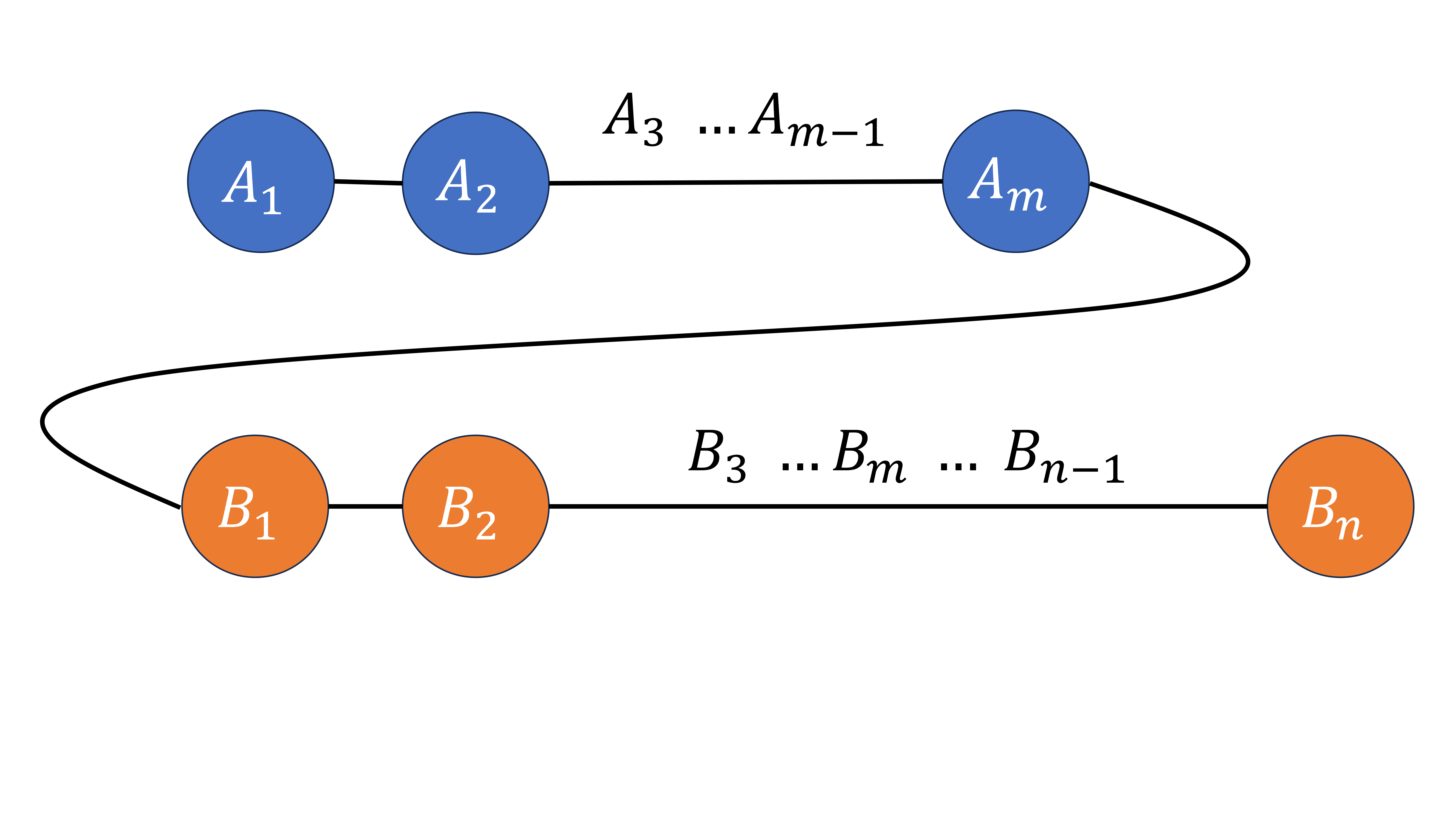}
        }
\\
\subfloat[The Kronecker product corresponding to Def. \ref{def:prods} (ii). It increases the local dimensions of the first $m$ subsystems.]
        {
        \label{fig:tensor2}
        \includegraphics[width=0.45\textwidth]{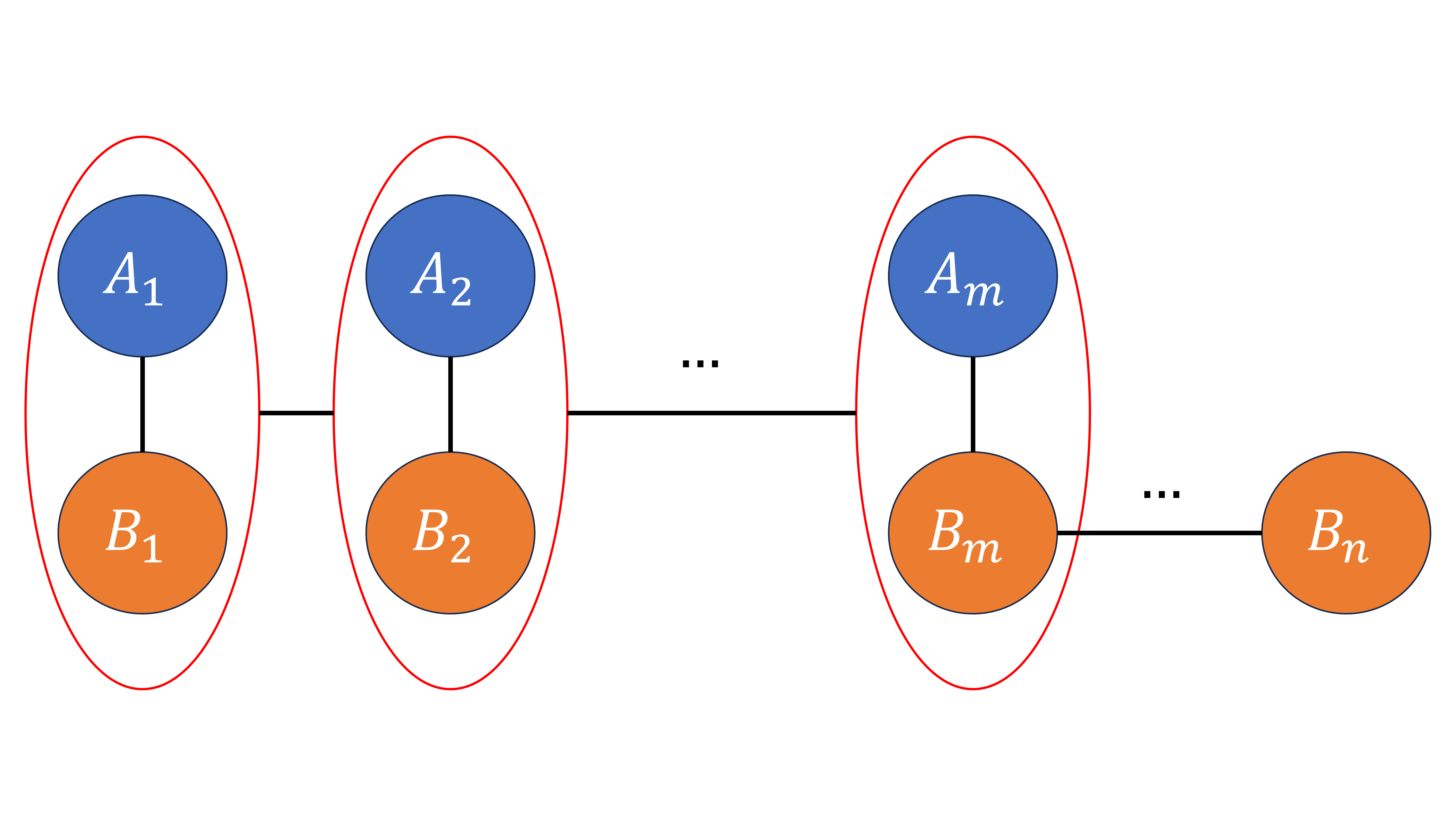}
        }
\\
\subfloat[The proposed $K_c$-product corresponding to Def. \ref{def:prods} (iii). It generates an $(m+n-l)$-partite system and increases a part of local dimensions.]
        {
        \label{fig:tensor3}
        \includegraphics[width=0.45\textwidth]{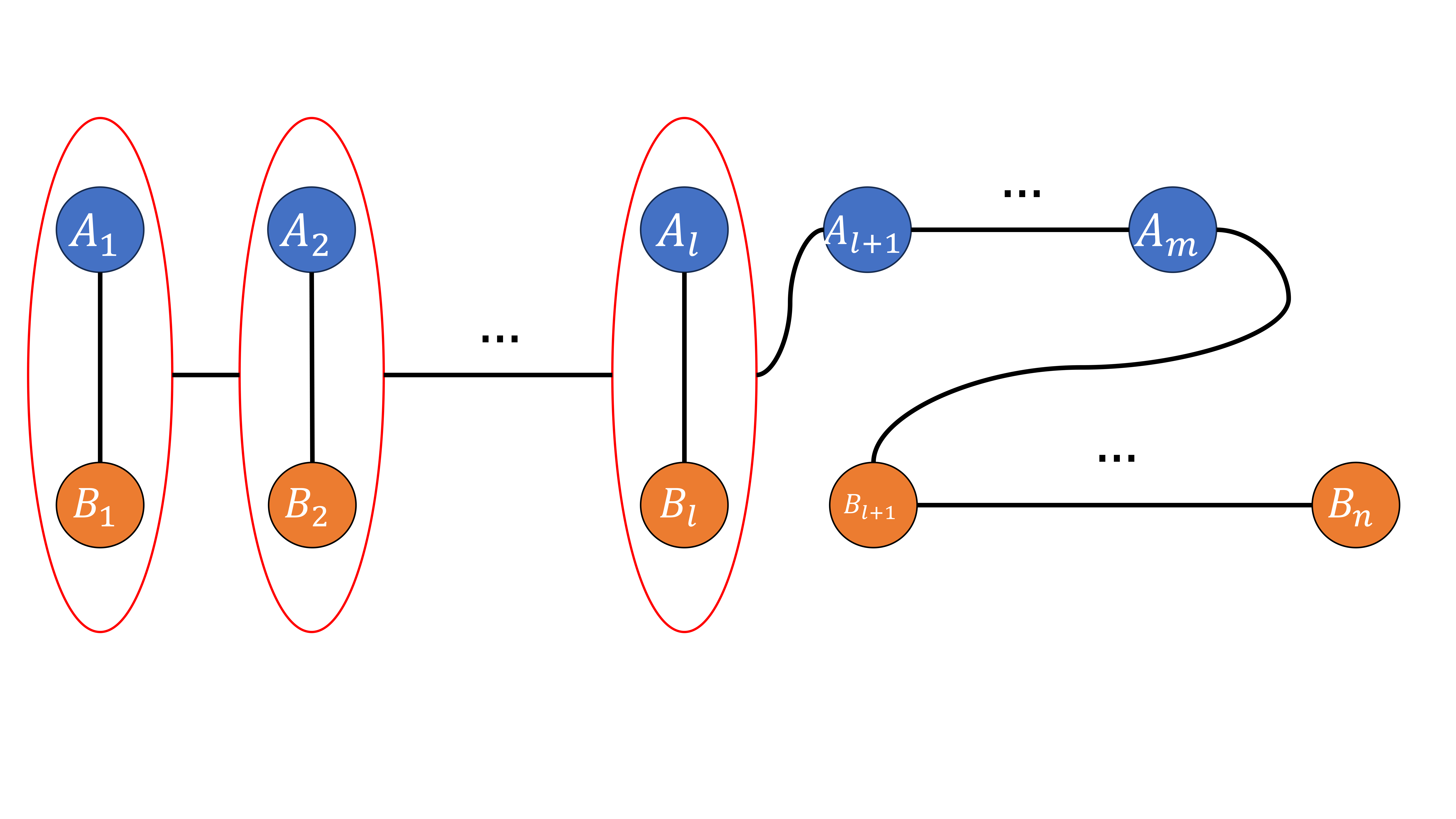}
        }
    \caption{Different composite ways of tensor product by Definition \ref{def:prods}.}
    \label{fig:tensors}
\end{figure}

\section{Properties of UDA states}
\label{sec:prop}

In this section we extend two essential properties of pure $k$-UDA states to mixed $k$-UDA states, i.e., the properties hold for arbitrary $k$-UDA states. First, it is known that the uniqueness among all states is maintained under local unitary (LU) equivalence for pure $k$-UDA states. In Lemma \ref{le:luequiv} we show that this uniqueness is also maintained under LU equivalence for mixed $k$-UDA states. Second, we derive the inclusion relation between the set of $k$-UDA states and that of $(k+1)$-UDA states in Lemma \ref{le:ktok+1}, in terms of arbitrary states.

\begin{lemma}
    \label{le:luequiv}
    If an $n$-partite (pure or mixed) state $\rho$ is $k$-UDA, then any state LU equivalent to $\rho$ is also $k$-UDA.
\end{lemma}

\begin{proof}
We prove it by contradiction. Suppose that $\rho$ is an $n$-partite $k$-UDA state. Next, assume that $\sigma$ is another $n$-partite state which is LU equivalent to $\rho$ but not $k$-UDA. It follows that $\sigma=(U_1\ox\cdots\ox U_n)\rho (U_1\ox\cdots\ox U_n)^\dg$ for some unitary operators $U_1,\cdots,U_n$. Due to the assumption, we conclude by definition that there exists an $n$-partite state $\alpha(\neq \sigma)$ all of whose $k$-partite marginals are the same as $\sigma$. Let $\tilde{\alpha}=(U_1\ox\cdots\ox U_n)^\dg \alpha (U_1\ox\cdots\ox U_n)$. It is obvious that $\tilde{\alpha}\neq \rho$. However, one can verify that for any subset $\cS\subset [n]$ with $\abs{\cS}=k$,
\begin{equation}
\label{eq:luequiv}
\begin{aligned}
\tilde{\alpha}_{\cS}&:=\tr_{\cS^c}(\tilde{\alpha}) \\
&=\tr_{\cS^c}\bigg((U_1\ox\cdots\ox U_n)^\dg \alpha (U_1\ox\cdots\ox U_n)\bigg) \\
&=\big(\bigox_{j\in\cS} U_j\big)^\dg \tr_{\cS^c}(\alpha) \big(\bigox_{j\in\cS} U_j\big) \\
&=\big(\bigox_{j\in\cS} U_j\big)^\dg \tr_{\cS^c}(\sigma) \big(\bigox_{j\in\cS} U_j\big) \\
&=\tr_{\cS^c}\bigg((U_1\ox\cdots\ox U_n)^\dg \sigma (U_1\ox\cdots\ox U_n)\bigg) \\
&=\tr_{\cS^c} (\rho)=\rho_{\cS}.
\end{aligned}
\end{equation}
It implies that $\tilde{\alpha}$ shares all the same $k$-partite marginals as $\rho$, and we obtain a contradiction. Thus, we conclude that every state that is LU equivalent to $\rho$ is also $k$-UDA. This completes the proof.
\end{proof}

The following lemma reveals the relation between the set of $k$-UDA states and the set of $(k+1)$-UDA states, in terms of arbitrary states.

\begin{lemma}
\label{le:ktok+1}
If $\rho$ is $k$-UDA, then $\rho$ is also $(k+1)$-UDA. However, the converse is generally incorrect.
\end{lemma}

\begin{proof}
We prove it by contradiction. Assume that $\rho$ is an $n$-partite state which is $k$-UDA but not $(k+1)$-UDA. Due to the assumption, it implies by definition that there is a state $\sigma(\neq \rho)$ having all the same $(k+1)$-partite marginals as $\rho$. That is $\rho_{\cS}=\sigma_{\cS}$ for any subset $\cS\subset [n]$ with $\abs{\cS}=k+1$. Further, any $k$-partite marginal can be generated from some $(k+1)$-partite marginal by tracing one more subsystem. It follows that
\begin{equation}
\label{eq:ktok+1-1}
\tr_{j\in\cS}(\rho_{\cS})=\tr_{j\in\cS}(\sigma_{\cS}), ~~\forall ~\cS.
\end{equation}
According to Eq. \eqref{eq:ktok+1-1} we conclude that $\sigma$ also has all the same $k$-partite marginals as $\rho$. It contradicts the assumption that $\rho$ is $k$-UDA, and thus $\rho$ have to be $(k+1)$-UDA. For the converse, one can verify it is incorrect by the following example. It is known that generic three-qubit pure states are $2$-UDA \cite{3qbuda} but cannot be fixed by their single-body marginals. This completes the proof.
\end{proof}

From Lemma \ref{le:ktok+1}, it implies that the set of $k$-UDA states is strictly included in the set of $(k+1)$-UDA states. We shall use Lemma \ref{le:ktok+1} to derive Corollary \ref{cor:mixed-k} in the next section.
Moreover, the marginals of fewer parties are easier to measure in experiments. Thus, from a experimental perspective, it is more valuable to know the $k$-UDA states with the number of parties $k$ as small as possible.

\section{The additivity of UDA states for different composite ways of tensor product}
\label{sec:add}

In this section we propose the concept of \emph{additivity} on $k$-UDA states, and study whether the additivity of $k$-UDA states holds for different composite ways of tensor product. For each composite way defined in Definition \ref{def:prods}, we call the \emph{additivity} of two $k$-UDA states if the corresponding product of them is still $k$-UDA.

We consider the additivity of $k$-UDA states for the following reasons. 
First, recall that the uniqueness by marginals depends on a fraction of parties sharing the marginal reductions \cite{udaubound,npudalbound}. As the number of all parties of the global state increases, only when the number of parties of the marginals is large enough can the uniqueness be ensured. Therefore, it is necessary to identify which states of large number of parties can be UDA by their $k$-partite marginals for $k$ much smaller than the fraction of all parties. One direction is to enlarge the number of parties of the global system while fixing the number $k$. In this view, we adopt the method of constructing GME states proposed in Ref. \cite{geshen2020} to generate $k$-UDA states in the systems of large number of parties. Furthermore, from an experimental perspective, measuring $k$-partite marginals for smaller $k$ is more practical. It also leads us to consider fixing the number of parties of the marginals. 
Second, by Definition \ref{def:prods} (i), we may expand the number of parties and enlarge local dimensions simultaneously via the proposed $K_c$-product. Hence, this product provides a tool to study the $k$-UDA states in the systems of distinct local dimensions. In some problems related to the correlation between the whole and parts, whether the local dimensions of a system are the same will result in essentail differences. For instance, the existence of absolutely maximally entangled states and $k$-uniform states in the systems of distinct local dimensions is quite different from that in the systems of equal local dimensions \cite{Shen2021,SF2022}.
Third, since the $K_c$-product is effective to construct GME states from Ref. \cite{geshen2020}, it is possible to construct multipartite $k$-UDA states with genuine multipartite entanglement by uniting such construction of GME states and the additivity of $k$-UDA states. Due to the extensive use of GME states, it is valuable to perform QST based on the genuinely entangled $k$-UDA states. We shall detailedly discuss the construction of genuinely entangled $k$-UDA states in Sec. \ref{sec:geuda}.


Based on Definition \ref{def:prods}, we specifically consider whether such three composite states $\rho\ox\sigma$, $\rho\ox_K\sigma$ and $\rho\ox_{K_c}\sigma$ are still $k$-UDA states for two $k$-UDA states $\rho$ and $\sigma$. We start from assuming that one of $\rho$ and $\sigma$ is pure. Under this assumption, we may explicitly formulate the expressions of the composite states as follows, under each type of tensor product.

\begin{lemma}
\label{le:preduced}
(i) If the reduction of system $(A_1,\cdots,A_m)$ from the global state $\rho_{A_1\cdots A_m E}$ is a pure state, then the global state is in the form as 
$$\rho_{A_1\cdots A_m E}=\proj{\psi}_{A_1\cdots A_m}\ox\sigma_E.$$ 

(ii) If the reduction of system $(A_1,\cdots,A_m)$ from an $n$-partite state $\rho_{(A_1B_1)\cdots(A_mB_m)B_{m+1}\cdots B_n}$ is a pure state, then the $n$-partite global state is in the form as 
$$\rho_{(A_1B_1)\cdots(A_mB_m)B_{m+1}\cdots B_n}=\proj{\psi}_{A_1\cdots A_m}\ox_K\gamma_{B_1\cdots B_n}.$$

(iii) If the reduction of system $(A_1,\cdots,A_m)$ from an $(m+n-l)$-partite state $\rho_{C_1\cdots C_l A_{l+1}\cdots A_m B_{l+1}\cdots B_n}$ is a pure state, where $C_j=A_jB_j$ for $1\leq j\leq l$ and $l<\min\{m,n\}$, then the $(m+n-l)$-partite global state is in the form as 
$$\rho_{C_1\cdots C_l A_{l+1}\cdots A_m B_{l+1}\cdots B_n}=\proj{\psi}_{A_1\cdots A_m}\ox_{K_c}\delta_{B_1\cdots B_n}.$$
\end{lemma}

\begin{proof}
(i) Suppose that the global state has the following decomposition:
\begin{equation}
\label{eq:preduced-1}
\rho_{A_1\cdots A_m E}=\sum_j\proj{x_j}_{A_1\cdots A_m E}. 
\end{equation}
Since the reduced state of system $(A_1,\cdots,A_m)$ is pure, we conclude that each $\ket{x_j}_{A_1\cdots A_m E}$ in Eq. \eqref{eq:preduced-1} is a product vector in the bipartition $(A_1\cdots A_m) | E$. Otherwise, the reduced state of system $(A_1,\cdots,A_m)$ is of rank greater than one. Thus, we may assume 
\beq
\label{eq:preduced-1a}
\rho_{A_1\cdots A_m E}=\sum_j \proj{\phi_j}_{A_1\cdots A_m}\ox\proj{\a_j}_E.
\eeq
By calculation, $\tr_E(\rho_{A_1\cdots A_mE})=\sum_j a_j\proj{\phi_j}_{A_1\cdots A_m}$, where $a_j=\braket{\a_j}{\a_j}$. This reduction is a pure state, and we may assume $\tr_E\rho_{A_1\cdots A_mE}=\proj{\psi}_{A_1\cdots A_m}$. It implies that each $\ket{\phi_j}$ is proportional to $\ket{\psi}$. Then we conclude that $\rho_{A_1\cdots A_mE}=\proj{\psi}_{A_1\cdots A_m}\ox\sigma_E$.

(ii) The proof is similar to assertion (i). Suppose that the global state has the following decomposition:
\begin{equation}
\label{eq:preduced-2}
\rho_{C_1\cdots C_m B_{m+1}\cdots B_n}=\sum_j\proj{y_j}_{C_1\cdots C_m B_{m+1}\cdots B_n},
\end{equation}
where $C_j:=A_jB_j$ for $1\leq j\leq m$. Since the reduced state of system $(A_1,\cdots,A_m)$ is a pure state whose rank is only one, we conclude that each $\ket{y_j}$ in Eq. \eqref{eq:preduced-2} is a Kronecker product of two vectors in system $(A_1, \cdots, A_m)$ and system $(B_1,\cdots, B_n)$ respectively. According to the similar discussion in (i), we derive that 
$$\rho_{(A_1B_1)\cdots(A_mB_m)B_{m+1}\cdots B_n}=\proj{\psi}_{A_1\cdots A_m}\ox_K\gamma_{B_1\cdots B_n}.$$

(iii) Since the $K_c$-product proposed in Definition \ref{def:prods} (iii) is a joint use of the tensor product and the Kronecker product, we similarly derive assertion (iii) according to the discussion on the first two assertions.

This completes the proof.
\end{proof}

By virtue of the essential expressions in Lemma \ref{le:preduced}, we can show that two $k$-UDA states admit the additivity under each type of tensor product in the case when one of the two initial states is pure.

\begin{theorem}
\label{thm:add-p+m}
Suppose that $\a$ and $\b$ are two $k$-UDA states of systems $(A_1,\cdots,A_m)$ and $(B_1,\cdots,B_n)$ respectively. If one of $\a$ and $\b$ is pure, then 

(i) $\a\ox\b$ is an $(m+n)$-partite $k$-UDA state of system $(A_1,\cdots,A_m,B_1,\cdots,B_n)$;

(ii) $\a\ox_K\b$ is an $\ell$-partite $k$-UDA state of system $(C_1,\cdots,C_{\ell})$, where $\ell=\max\{m,n\}$ and $C_i:=(A_iB_i)$ for $i=1,\cdots,\ell$.

(iii) $\a\ox_{K_c}\b$ is an $(m+n-\ell)$-partite $k$-UDA state of system $(C_1,\cdots,C_{\ell},A_{\ell+1},\cdots,A_m,B_{\ell+1},\cdots,n)$, where $\ell\leq\min\{m,n\}$ and $C_i:=(A_iB_i)$ for $i=1,\cdots,\ell$.
\end{theorem}

\begin{proof}
First of all, we may assume that $\a_{A_1\cdots A_m}=\proj{\psi}$ without loss of generality. Denote by $A_{[m]}$ the $m$-partite system $(A_1,\cdots, A_m)$, and similarly by $B_{[n]}$ the $n$-partite system $(B_1,\cdots, B_n)$.

(i) Let $\rho_{A_{[m]}B_{[n]}}=\a_{A_1\cdots A_m}\ox\b_{B_1\cdots B_n}$. Suppose that $\sigma_{A_{[m]}B_{[n]}}$ is an $(m+n)$-partite state of system $(A_{[m]},B_{[n]})$, which shares all the same $k$-partite marginals as $\rho_{A_{[m]}B_{[n]}}$. Let $\sigma_{A_{[m]}}=\tr_{B_{[n]}}(\sigma_{A_{[m]}B_{[n]}})$. We claim that $\sigma_{A_{[m]}}$ shares all the same $k$-partite marginals as $\a_{A_{[m]}}$ for the following reason. According to the assumption, we obtain that for any subset $\cS\subset [m]$ with $\abs{\cS}=k$,
\begin{equation}
\label{eq:add-p+m-1}
\begin{aligned}
\sigma_{A_{\cS}}&:=\tr_{A_{\cS^c}} (\sigma_{A_{[m]}}) \\
&~=\tr_{A_{\cS^c}}\left[\tr_{B_{[n]}}(\sigma_{A_{[m]}B_{[n]}})\right] \\
&~\equiv \tr_{A_{\cS^c}B_{[n]}}(\sigma_{A_{[m]}B_{[n]}}) \\
&~=\tr_{A_{\cS^c}B_{[n]}}(\rho_{A_{[m]}B_{[n]}}) \\
&~=\tr_{A_{\cS^c}}(\a_{A_{[m]}})\equiv\a_{A_{\cS^c}}.
\end{aligned}
\end{equation}
Recall that $\a_{A_{[m]}}$ is a $k$-UDA state. It follows from Eq. \eqref{eq:add-p+m-1} that $\sigma_{A_{[m]}}$ has to be equal to $\a_{A_{[m]}}$. 

Since $\a_{A_{[m]}}$ is pure, it follows from Lemma \ref{le:preduced} (i) that 
\beq
\label{eq:add-p+m-1a}
\sigma_{A_{[m]}B_{[n]}}=\proj{\psi}_{A_{[m]}}\ox\gamma_{B_{[n]}}.
\eeq
Similar to the discussion above, we also claim that $\gamma_{B_{[n]}}$ shares all the same $k$-partite marginals as $\b_{B_{[n]}}$. Since $\b_{B_{[n]}}$ is also $k$-UDA, it follows that $\gamma_{B_{[n]}}$ has to be equal to $\b_{B_{[n]}}$. It follows from Eq. \eqref{eq:add-p+m-1a} that $\sigma_{A_{[m]}B_{[n]}}$ is the tensor product of $\a_{A_{[m]}}$ and $\b_{B_{[n]}}$, which means $\sigma_{A_{[m]}B_{[n]}}$ must be identical to $\rho_{A_{[m]}B_{[n]}}$. Thus, by definition $\rho_{A_{[m]}B_{[n]}}$ is uniquely determined by its $k$-partite marginals.

(ii) The proof is similar to that of assertion (i). Without loss of generality, we may assume $m\leq n$, i.e. $n=\ell=\max\{m,n\}$. Let $\rho_{C_{[n]}}=\a_{A_{[m]}}\ox_K\b_{B_{[n]}}$, where the composite system $C_j=(A_j B_j)$ for $1\leq j \leq m$ and $C_j=B_j$ for $m+1\leq j\leq n$ for simplicity.
Denote by $\rho_{\cS}$ the reduced state of system $C_{\cS},~\forall~\cS\subset[n]$. Suppose that $\sigma_{C_{[n]}}$ is an $n$-partite state compatible with the marginal set $\{\rho_{\cS}~|~\forall \cS,~ \abs{\cS}=k\}$. Let $\sigma_{A_{[m]}}=\tr_{B_{[n]}}(\sigma_{C_{[n]}})$. We claim that $\sigma_{A_{[m]}}$ shares all the same $k$-partite marginals as $\a_{A_{[m]}}$ for the following reason. According to the assumption, for any subset $\tilde{\cS}\subset [m]$ with $\abs{\tilde{\cS}}=k$ and its complementary set $\tilde{\cS}^c=[m]-\tilde{\cS}$, we obtain that
\begin{equation}
\label{eq:add-p+m-2}
\begin{aligned}
\sigma_{A_{\tilde{\cS}}}&:=\tr_{A_{\tilde{\cS}^c}}(\sigma_{A_{[m]}}) \\
&~=\tr_{A_{\tilde{\cS}^c}}\left[\tr_{B_{[n]}}(\sigma_{C_{[n]}})\right] \\
&~=\tr_{B_{\tilde{\cS}}}\left[\tr_{C_{\tilde{\cS}^c} C_{[n\backslash m]}}(\sigma_{C_{[n]}})\right] \\
&~=\tr_{B_{\tilde{\cS}}}\left[\tr_{C_{\tilde{\cS}^c} C_{[n\backslash m]}}(\rho_{C_{[n]}})\right] \\
&~=\tr_{A_{\tilde{\cS}^c}}\left[\tr_{B_{[n]}}(\rho_{C_{[n]}})\right] \\
&~=\tr_{A_{\tilde{\cS}^c}}(\a_{A_{[m]}})\equiv\a_{A_{\tilde{\cS}}},
\end{aligned}
\end{equation}
where $[n\backslash m]:=[n]-[m]=\{m+1,\cdots,n\}$.
Recall that $\a_{A_{[m]}}$ is $k$-UDA. It follows from Eq. \eqref{eq:add-p+m-2} that $\sigma_{A_{[m]}}$ has to be equal to $\a_{A_{[m]}}$. 

Since $\sigma_{A_{[m]}}\equiv\a_{A_{[m]}}=\proj{\psi}$ is pure, it follows from Lemma \ref{le:preduced} (ii) that 
\beq
\label{eq:add-p+m-2a}
\sigma_{C_{[n]}}=\proj{\psi}_{A_{[m]}}\ox_K\gamma_{B_{[n]}}. 
\eeq
Similar to the discussion above, we claim that $\gamma_{B_{[n]}}$ shares all the same $k$-partite marginals as $\b_{B_{[n]}}$. Since $\b_{B_{[n]}}$ is also $k$-UDA, we obtain that $\gamma_{B_{[n]}}=\b_{B_{[n]}}$. It follows from Eq. \eqref{eq:add-p+m-2a} that $\sigma_{C_{[n]}}$ has to be the Kronecker product of $\a_{A_{[m]}}$ and $\b_{B_{[n]}}$, which means $\sigma_{C_{[n]}}$ must be identical to $\rho_{C_{[n]}}$. Thus, by definition $\rho_{C_{[n]}}$ is uniquely determined by its $k$-partite marginals.

(iii) For any $\ell\leq\min\{m,n\}$, let 
$$\rho_{C_{[\ell]}A_{[m\backslash \ell]}B_{[n\backslash \ell]}}=\alpha_{[m]}\ox_{K_c}\beta_{[n]},$$ 
where the composite system $C_j=A_j B_j$ for $1\leq j\leq \ell$, and $[m\backslash \ell]:=[m]-[\ell]$, $[n\backslash \ell]:=[n]-[\ell]$. Suppose that $\sigma_{C_{[\ell]}A_{[m\backslash \ell]}B_{[n\backslash \ell]}}$ is an $(m+n-\ell)$-partite state compatible with all the $k$-partite marginals from $\rho_{C_{[\ell]}A_{[m\backslash \ell]}B_{[n\backslash \ell]}}$.
Let $\sigma_{A_{[m]}}=\tr_{B_{[n]}}(\sigma_{C_{[\ell]}A_{[m\backslash \ell]}B_{[n\backslash \ell]}})$. We claim that $\sigma_{A_{[m]}}$ shares all the same $k$-partite marginals as $\alpha_{A_{[m]}}$ for the following reason. One can verify that for any subset $S\subset [m]$ with $\abs{S}=k$,
\begin{equation}
\label{eq:add-p+m-3}
\begin{aligned}
\sigma_{A_{\cS}}&:=\tr_{A_{\cS^c}}(\sigma_{A_{[m]}}) \\
&~=\tr_{A_{\cS^c}}\left[\tr_{B_{[n]}}(\sigma_{C_{[\ell]}A_{[m\backslash \ell]}B_{[n\backslash \ell]}})\right]\\
&~=\tr_{C_{\cT}A_{\cS^c-\cT}B_{[n]-\cT}}(\sigma_{C_{[\ell]}A_{[m\backslash \ell]}B_{[n\backslash \ell]}}) \\
&~=\tr_{B_{\cT^c}}(\sigma_{C_{\cT^c}A_{\cS-\cT^c}}),
\end{aligned}
\end{equation}
where $\cT=\cS^c\cap [\ell]$, and $\cT^c$ is the complementary set of $\cT$ in $[\ell]$.
Since $\abs{\cT^c}+\abs{\cS-\cT^c}=\abs{\cS}=k$, it implies that $\sigma_{C_{\cT^c}A_{\cS-\cT^c}}$ is a $k$-partite marginal of the $(m+n-\ell)$-partite state $\sigma_{C_{[\ell]}A_{[m\backslash \ell]}B_{[n\backslash \ell]}}$. According to the assumption that $\sigma_{C_{[\ell]}A_{[m\backslash \ell]}B_{[n\backslash \ell]}}$ shares all the same $k$-partite marginals as $\rho$, it follows from Eq. \eqref{eq:add-p+m-3} that 
\begin{equation}
\label{eq:add-p+m-3a}
\begin{aligned}
\sigma_{A_{\cS}}&=\tr_{B_{\cT^c}}(\rho_{C_{\cT^c}A_{\cS-\cT^c}})\\
&=\alpha_{A_{\cS}}, \quad \forall \abs{\cS}=k.
\end{aligned}
\end{equation}
Recall that $\alpha_{A_{[m]}}$ is $k$-UDA. It follows from Eq. \eqref{eq:add-p+m-3a} that $\sigma_{A_{[m]}}$ has to be equal to $\alpha_{A_{[m]}}$. 

Since $\alpha_{A_{[m]}}=\proj{\psi}$ is pure, it follows from Lemma \ref{le:preduced} (iii) that
\begin{equation}
\label{eq:add-p+m-3b}
\sigma_{C_{[\ell]}A_{[m\backslash \ell]}B_{[n\backslash \ell]}}=\proj{\psi}_{A_{[m]}}\ox_{K_c}\delta_{B_{[n]}}.
\end{equation}
Similar to the above discussion, we claim that $\delta_{B_{[n]}}$ shares all the same $k$-partite marginals as $\beta_{B_{[n]}}$. Since $\beta_{B_{[n]}}$ is also $k$-UDA, we obtain that $\delta_{B_{[n]}}=\beta_{B_{[n]}}$, and $\sigma_{C_{[\ell]}A_{[m\backslash \ell]}B_{[n\backslash \ell]}}$ must be identical to $\rho_{C_{[\ell]}A_{[m\backslash \ell]}B_{[n\backslash \ell]}}$ from Eq. \eqref{eq:add-p+m-3b}.
Thus, by definition $\rho_{C_{[\ell]}A_{[m\backslash \ell]}B_{[n\backslash \ell]}}$ is uniquely determined by its $k$-partite marginals.

This completes the proof.
\end{proof}

Based on Theorem \ref{thm:add-p+m} and its proof, we obtain two direct corollaries as follows. First, recall that the $k$-UDA states are typically defined on pure states. Therefore, Theorem \ref{thm:add-p+m} reveals that the additivity of $k$-UDA states corresponding to each type of tensor product holds for the typical definition, i.e. Definition \ref{def:udp-uda}. Second, according to the proof, we observe that the composite states $\a\ox\b$, $\a\ox_{K}\b$ and $\a\ox_{K_c}\b$ can be completely determined by only a part of its $k$-partite marginals. Take $\a\ox\b$ in Theorem \ref{thm:add-p+m} (i) as an example. It follows from Eqs. \eqref{eq:add-p+m-1} and \eqref{eq:add-p+m-1a} that the composite state $\a\ox\b$ can be fixed, as long as the $k$-partite marginals of systems $A_{\cS}$ and $B_{\cT}$ can uniquely determine the two parts $\a$ and $\b$ respectively. It implies that the $k$-partite marginals of systems $(A_{\cS}B_{\cT})$ are redundant, for $\cS\subset[m]$, $\cT\subset[n]$ and $\abs{\cS}+\abs{\cT}=k$. We obtain similar conclusions for Theorem \ref{thm:add-p+m} (i) and (ii).

Moreover, in Theorem \ref{thm:add-p+m}, the two initial states $\a$ and $\b$ are both $k$-UDA. Here, by virtue of Lemma \ref{le:ktok+1}, we extend Theorem \ref{thm:add-p+m} to the case when $\a$ and $\b$ are $k_1$-UDA and $k_2$-UDA respectively for different $k_1$ and $k_2$.


\begin{corollary}
\label{cor:mixed-k}
Suppose $\a$ is a $k_1$-UDA state of system $(A_1,\cdots,A_m)$, and $\b$ is a $k_2$-UDA state of system $(B_1,\cdots,B_n)$. Let $k=\max\{k_1,k_2\}$. If one of $\a$ and $\b$ is a pure state, then

(i) $\a\ox\b$ is an $(m+n)$-partite $k$-UDA state of system $(A_1,\cdots,A_m,B_1,\cdots,B_n)$;

(ii) $\a\ox_K\b$ is an $\ell$-partite $k$-UDA state of system $(C_1,\cdots,C_{\ell})$, where $\ell=\max\{m,n\}$ and $C_i:=A_iB_i$ for $i=1,\cdots,\ell$.

(iii) $\a\ox_{K_c}\b$ is an $(m+n-\ell)$-partite $k$-UDA state of system $(C_1,\cdots,C_{\ell},A_{\ell+1},\cdots,A_m,B_{\ell+1},\cdots,n)$, where $\ell\leq\min\{m,n\}$ and $C_i:=(A_iB_i)$ for $i=1,\cdots,\ell$.
\end{corollary}

Next, we consider the case when the initial two states are both mixed $k$-UDA states. We provide a characterization of the states that are compatible with all $k$-partite marginals of the composite state for different types of tensor product. Since the $K_c$-product is a combination of the tensor product and the Kronecker product by applying the Kronecker product on a part of subsystems, we shall consider the two fundamental products for simplicity.

\begin{lemma}
\label{le:prodadd}
Suppose that $\rho$ and $\sigma$ are both $n$-partite $k$-UDA states supported on the Hilbert spaces $\cH_{A_1}\ox \cdots\ox \cH_{A_n}$ and $\cH_{B_1}\ox \cdots\ox \cH_{B_n}$ respectively. Then,

(i) the states compatible with all $k$-partite marginals of $\rho\ox \sigma$ admit an expansion as $\rho\ox \sigma+\chi$, where the correlation matrix $\chi_{A_{[n]}B_{[n]}}$ has to be trace zero, and satisfies that (a) the reductions $\chi_{A_{[n]}}, ~\chi_{B_{[n]}}$ are both zero; (b) any $k$-partite reduction of $\chi_{A_{[n]}B_{[n]}}$ is zero.

(ii) the states compatible with all $k$-partite marginals of $\rho\ox_K \sigma$ admit an expansion as $\rho\ox_K \sigma+\gamma$, where the $n$-partite correlation matrix $\gamma_{(AB)_{[n]}}$ has to be trace zero, and satisfies that (a) the reductions $\gamma_{A_{[n]}},~\gamma_{B_{[n]}}$ are both zero; (b) any $k$-partite reduction of $\gamma_{(AB)_{[n]}}$ is zero.
\end{lemma}

\begin{proof}
(i) Assume that $\a$ is a $2n$-partite state supported on the corresponding Hilbert space, which shares all the same $k$-partite marginals as $\rho\ox\sigma$. It follows that
\begin{equation}
\label{eq:tprod-1}
\begin{aligned}
\a_{A_{\cS}}&:=\tr_{A_{\cS^c}}\big[\tr_{B_{[n]}}(\a)\big] \\
&~\equiv\tr_{A_{\cS^c}}\big[\tr_{B_{[n]}}(\rho\ox\sigma)\big]=\rho_{A_{\cS}}, \\
\a_{B_{\cS}}&:=\tr_{B_{\cS^c}}\big[\tr_{A_{[n]}}(\a)\big]\\
&~\equiv\tr_{B_{\cS^c}}\big[\tr_{A_{[n]}}(\rho\ox\sigma)\big]=\sigma_{B_{\cS}},
\end{aligned}
\end{equation}
for any subset $\cS\subset[n]$ with $\abs{\cS}=k$. That is, the two sets of marginals $\big\{\a_{A_{\cS}}~|~\forall~\abs{\cS}=k\big\}$ and $\big\{\rho_{A_{\cS}}~|~\forall~\abs{\cS}=k\big\}$ are identical, and the two sets of marginals $\big\{\a_{B_{\cS}}~|~\forall~\abs{\cS}=k\big\}$ and $\big\{\sigma_{B_{\cS}}~|~\forall~\abs{\cS}=k\big\}$ are identical.
Since $\rho$ and $\sigma$ are both $k$-UDA, it implies that the set of marginals $\big\{\a_{A_{\cS}}~|~\forall~\abs{\cS}=k\big\}$ is only compatible with the $n$-partite state $\rho$, and the set of marginals $\big\{\a_{B_{\cS}}~|~\forall~\abs{\cS}=k\big\}$ is only compatible with the $n$-partite state $\sigma$. It means that
\begin{equation}
\label{eq:tprod-2}
\bal
\a_{A_1,\cdots,A_n}&:=\tr_{B_1,\cdots,B_n}(\a)=\rho, \\
\a_{B_1,\cdots,B_n}&:=\tr_{A_1,\cdots,A_n}(\a)=\sigma.
\eal
\end{equation}
Then we may regard $\a$ as a bipartite state of the system $(A_{[n]},(B_{[n]})$ whose two reduced density matrices are exactly $\rho$ and $\sigma$ respectively. Thus, the generic expansion of $\a$ is $\rho\ox\sigma+\chi$, where the correlation matrix $\chi$ has to be trace zero, and satisfies that $\chi_{A_{[n]}}=\chi_{B_{[n]}}=0$ \cite{trimarginal2014}. Moreover, since $\a$ is compatible with all $k$-partite marginals of $\rho\ox\sigma$, it requires that any $k$-partite reduction of $\chi$ is zero.

(ii) Assume that $\a$ is an $n$-partite state supported on the corresponding Hilbert space, which shares all the same $k$-partite marginals as $\rho\ox_K\sigma$. It follows that
\begin{equation}
\label{eq:kprod-1}
\a_{(AB)_{\cS}}:=\tr_{(AB)_{\cS^c}} (\a)\equiv\tr_{(AB)_{\cS^c}} (\rho\ox_K\sigma),
\end{equation}
for any subset $\cS\subset[n]$ with $\abs{\cS}=k$. It follows that
\begin{equation}
\label{eq:kprod-2}
\begin{aligned}
\a_{A_{\cS}}&:=\tr_{B_{\cS}}(\a_{(AB)_{\cS}}) \\
&~\equiv \tr_{B_{\cS}}\left[\tr_{(AB)_{\cS^c}} (\rho\ox_K\sigma)\right]=\rho_{A_{\cS}}, \\
\a_{B_{\cS}}&:=\tr_{A_{\cS}}(\a_{(AB)_{\cS}}) \\
&~\equiv \tr_{A_{\cS}}\left[\tr_{(AB)_{\cS^c}} (\rho\ox_K\sigma)\right]=\sigma_{B_{\cS}}, 
\end{aligned}
\end{equation}
for any subset $\cS\subset[n]$ with $\abs{\cS}=k$. Since $\rho$ and $\sigma$ are both $k$-UDA, it implies that the set of marginals $\big\{\a_{A_{\cS}}~|~\forall~\abs{\cS}=k\big\}$ is only compatible with the $n$-partite state $\rho$, and the set of marginals $\big\{\a_{B_{\cS}}~|~\forall~\abs{\cS}=k\big\}$ is only compatible with the $n$-partite state $\sigma$. It means that
\begin{equation}
\label{eq:kprod-3}
\bal
\a_{A_1,\cdots,A_n}:=\tr_{B_1,\cdots,B_n}(\a)=\rho, \\
\a_{B_1,\cdots,B_n}:=\tr_{A_1,\cdots,A_n}(\a)=\sigma.
\eal
\end{equation}
Then we similarly obtain the generic expansion of $\a$ as $\rho\ox_K\sigma+\gamma$, where the correltaion matrix $\gamma$ has to be trace zero, and satisfies that $\gamma_{A_{[n]}}=\gamma_{B_{[n]}}=0$. Moreover, since $\a$ shares all the same $k$-partite marginals as $\rho\ox_K\sigma$, it implies that any $k$-partite marginal of $\gamma$ is zero, i.e., $\gamma_{(AB)_{\cS}}=0$ for any $\abs{\cS}=k$.

This completes the proof.
\end{proof}

By observation on Lemma \ref{le:prodadd}, the correlation matrices $\chi$ and $\gamma$ are essential to characterize the states which are compatible with the marginals of $\rho\ox\sigma$ and $\rho\ox_K\sigma$ respectively. It is direct to observe from Lemma \ref{le:prodadd} that $\rho\ox\sigma$ and $\rho\ox_K\sigma$ are $k$-UDA if and only if the corresponding correlation matrices $\chi$ and $\gamma$ under the constraints must be zero matrices. 
Hence, to determine the uniquess, it is necessary to further study the existence of the two correlation matrices with required conditions. We specifically analyze the necessary condition that any $k$-partite reduction is zero for both multipartite correlation matrices $\chi$ and $\gamma$ in Lemma \ref{le:prodadd}.
We find that this condition cannot restrict a multipartite Hermitian matrix of zero trace to be zero. For example, the following is a Hermitian matrix of zero trace acting on $\mathbb{C}^2\ox\mathbb{C}^2$, each of whose single-body marginals is zero:
\begin{equation}
\label{eq:twoqubit-cormat}
\bma
m_1 & m_2 & m_3 & m_4 \\
m_2^* & -m_1 & m_5 & -m_3 \\
m_3^* & m_5^* & -m_1 & -m_2 \\
m_4^* & -m_3^* & -m_2^* & m_1 
\ema
\end{equation}
for some real $m_1$ and complex $m_2,m_3,m_4,m_5$. This Hermitian matrix is nonzero if one of $m_1,\cdots,m_5$ are nonzero.
Next, we show the general result on multipartite Hermitian matrices of zero trace. 

\begin{lemma}
\label{le:ptzero}
There exist infinitely many non-zero multipartite Hermitian matrices of zero trace whose $k$-partite reductions are all zero.
\end{lemma}

\begin{proof}
Suppose that $\chi$ is an $n$-partite Hermitian matrix of zero trace acting on the Hilbert space $\cH_{A_1}\ox\cdots\ox\cH_{A_n}$, and each $k$-partite reduction of $\chi$ is zero. We formulate the matrix form of $\chi$ as
\begin{equation}
\label{eq:pt=zero}
\chi=\sum_{\substack{i_1,\cdots,i_n\\ j_1,\cdots, j_n}} m_{\substack{i_1,\cdots,i_n\\ j_1,\cdots, j_n}}\ketbra{i_1,\cdots,i_n}{j_1,\cdots, j_n},
\end{equation}
where 
\begin{equation*}
\bal
&\{\ket{i_l}:\ket{0},\cdots,\ket{d_l-1}\} \\
&\{\ket{j_l}:\ket{0},\cdots,\ket{d_l-1}\}
\eal
\end{equation*}
are both the computational basis of the $l$-th subspace $\cH_{A_l}$ for any $l=1,\cdots,n$. Also the matrix elements satisfy that $m_{\substack{i_1,\cdots,i_n\\ j_1,\cdots, j_n}}=m^*_{\substack{j_1,\cdots,j_n\\ i_1,\cdots, i_n}}$ because $\chi$ is Hermitian.
For any subset $\cS\subset[n]$ with $\abs{\cS}=k$, the $k$-partite reduction $\chi_{A_{\cS}}$ can be calculated as 
\begin{equation}
\label{eq:pt=zero-1}
\begin{aligned}
\chi_{A_{\cS}}&:=\tr_{A_{\cS^c}}(\chi) \\
&~=\sum_{\substack{i_1,\cdots,i_n\\ j_1,\cdots, j_n}} \braket{j_{\cS^c}}{i_{\cS^c}} m_{\substack{i_1,\cdots,i_n\\ j_1,\cdots, j_n}} \a_{ij}^{s_1,\cdots,s_k} \\
&~=\sum_{\substack{i_{s_1},\cdots,i_{s_k}\\ j_{s_1},\cdots,j_{s_k}}} \big(\sum_{i_{\cS^c}=j_{\cS^c}} m_{\substack{i_1,\cdots,i_n\\ j_1,\cdots, j_n}}\big) \a_{ij}^{s_1,\cdots,s_k},
\end{aligned}
\end{equation}
where $\a_{ij}^{s_1,\cdots,s_k}=\ketbra{i_{s_1},\cdots,i_{s_k}}{j_{s_1},\cdots,j_{s_k}}$,
$s_1,\cdots,s_k\in\cS$, $i_{\cS^c}$ denotes a tuple $(i_{t_1},\cdots,i_{t_{n-k}})$ for $t_1,\cdots,t_{n-k}\in\cS^c$, and similarly for $j_{\cS^c}$.
Due to the assumption that $\chi_{A_{\cS}}=0$ for any subset $\cS$, from Eq. \eqref{eq:pt=zero-1} we obtain that for any subset $\cS$, each sum as $\sum\limits_{i_{\cS^c}=j_{\cS^c}} m_{\substack{i_1,\cdots,i_n\\ j_1,\cdots, j_n}}$ is zero. This requirement cannot restrict Hermitian $\chi$ to be a zero matrix. For example, the following is such a non-zero Hermitian matrix:
\begin{equation}
\label{eq:pt=zero-1.1}
\begin{aligned}
\chi=&\sum_{\forall l,i_l\neq j_l} \bigg(c_{\substack{i_1,\cdots,i_n\\ j_1,\cdots,j_n}}\ketbra{i_1,\cdots,i_n}{j_1,\cdots,j_n} \\
&\qquad\quad+c^*_{\substack{i_1,\cdots,i_n\\ j_1,\cdots,j_n}}\ketbra{j_1,\cdots,j_n}{i_1,\cdots,i_n}\bigg).
\end{aligned}
\end{equation}
One can verify that each $(n-1)$-partite reduction of the Hermitian $\chi$ in Eq. \eqref{eq:pt=zero-1.1} is zero, and thus each $k$-partite reduction is zero for any $k<n$. 
\end{proof}

Lemma \ref{le:ptzero} indicates that there could be nonzero correlation matrices $\chi$ and $\gamma$ satisfying the conditions given in Lemma \ref{le:prodadd} such that $\rho\ox\sigma+\chi$ and $\rho\ox_K\sigma+\gamma$ are positive semidefinite. Therefore, we conjecture there is generally no additivity when the two initial states are both mixed $k$-UDA states, for different composite ways of tensor product.

\section{Construction of UDA states with genuine multipartite entanglement}
\label{sec:geuda}

The GME states are resourceful, which have been widely used in various quantum information processing tasks. Thus, it is valuable to perform QST of GME states from an experimental perspective. In this view, it is necessary to consider the uniqueness issue on the GME states. It is also connected to detecting multipartite entanglement. Note that graph states are entangled pure states that exhibit complex structures of genuine multipartite entanglement. In Ref. \cite{end2uda2010} the authors consider detecting graph-state entanglement by measuring two-particle correlations only, and concluded that this is impossible. For this reason, in this section we specifically study the $2$-UDA states with genuine multipartite entanglement in the systems of large number of parties and high local dimensions. In other words, for such states, genuine multipartite entanglement can be detected by measuring two-particle correlations only. 

Inspired by the additivity of $k$-UDA states derived in Theorem \ref{thm:add-p+m} and the construction of GME states in Ref. \cite{geshen2020,Chen_2023}, we find that it is effective to construct genuinely entangled $k$-UDA states via the Kronecker product and the $K_c$-product given in Definition \ref{def:prods} (ii) and (iii) respectively. One can verify that the Kronecker product of two GME states is also genuinely entangled. For the $K_c$-product, we have shown that if the range of one of the two initial states is not spanned by biproduct vectors, then the output state via $K_c$-product must be a GME state supported on the corresponding Hilbert space \cite{geshen2020}. 
As a special but important case, the range of pure GME states cannot be spanned by biproduct vectors. Hence, if one of the two input states is a pure GME state, then the output state via the $K_c$-product must be genuinely entangled by the statement above-mentioned. Then combining such a construction with the additivity of $k$-UDA states, we derive an effective way to construct genuinely entangled $k$-UDA states as follows.

\begin{proposition}
\label{prop:geuda}
Suppose that $\a$ and $\b$ are two $k$-UDA states of systems $(A_1,\cdots,A_m)$ and $(B_1,\cdots,B_n)$ respectively, and they are both genuinely entangled. If one of $\a$ and $\b$ is pure, then 

(i) $\a\ox_K\b$ is an $\ell$-partite genuinely entangled $k$-UDA state, where $\ell=\max\{m,n\}$.

(ii) $\a\ox_{K_c}\b$ is an $(m+n-\ell)$-partite genuinely entangled $k$-UDA state, for any $\ell<\min\{m,n\}$.
\end{proposition}

From Proposition \ref{prop:geuda} we can generate genuinely entangled $k$-UDA states in the systems of more parties and higher local dimensions using two genuinely entangled $k$-UDA states. By repeating the composite process, the number of parties and local dimensions can be continuously increased. We illustrate the construction process given by Proposition \ref{prop:geuda} in Fig. \ref{fig:geuda}.

\begin{figure}[htbp]
\centering
\includegraphics[width=0.48\textwidth]{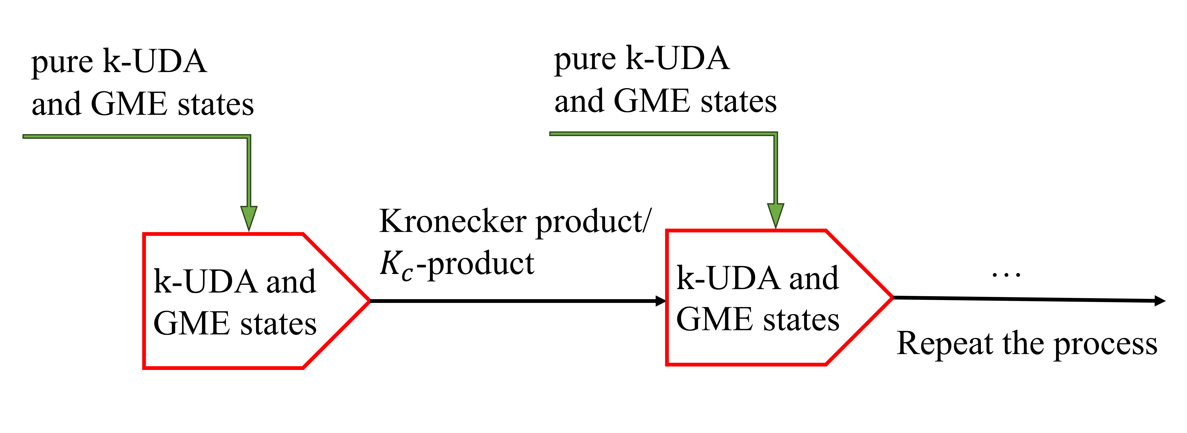}
\caption{The process of constructing genuinely entangled $k$-UDA states of large number of parties. From Proposition \ref{prop:geuda}, the output state of two genuinely entangled $k$-UDA states (one of which should be pure) via the Kronecker product or the $K_c$-product is still $k$-UDA and genuinely entangled. By repeating the composition of a pure genuinely entangled $k$-UDA state and the output state, the number of parties and local dimensions can be continuously increased.}
\label{fig:geuda}
\end{figure}

Next, we propose a class of GME states which are $2$-UDA by virtue of Proposition \ref{prop:geuda}. According to the construction process illustrated by Fig. \ref{fig:geuda}, it is fundamental to discover two genuinely entangled $2$-UDA states as the initial states. First, it is known that only three-qubit generalized GHZ states and their local unitary equivalents cannot be uniquely determined by the two-qubit marginal reductions \cite{GGHZ2009}. Further, it was shown in Ref. \cite{3qubitc2001} that there are exactly two locally inequivalently classes of genuinely entangled pure three-qubit states, namely the class of GHZ type states and the class of $W$ type states. Hence, the three-qubit $W$ type states are both genuinely entangled and $2$-UDA. Second, there is a mixed tripartite state constructed in Ref. \cite{trimarginal2014} which is both genuinely entangled and $2$-UDA. Based on the facts above, we propose the following class of states.


\begin{example}
\label{ex:geuda-1}
The standard, unique form of three-qubit $W$ type was derived in Ref. \cite{ent3qubi2000} as 
\beq
\label{eq:wtype}
\ket{\psi_{W}}=\sqrt{a}\ket{001}+\sqrt{b}\ket{010}+\sqrt{c}\ket{100}+\sqrt{d}\ket{000},
\eeq
where $a,b,c>0$, and $d\equiv1-(a+b+c)\geq 0$. It is also known from Ref. \cite{ent3qubi2000} that $\ket{\psi_{W}}$ is genuinely entangled and cannot be converted to the generalized GHZ states by invertible local operators, and thus $2$-UDA.

Moreover, the following mixed states supported on ${\cal H}_{B_1B_2B_3}\cong\mathbb C^{d+1}\ox\mathbb C^{d+1}\ox\mathbb C^{d+1}$ are genuinely entangled and can be uniquely determined by their bipartite marginals \cite{trimarginal2014}:
\begin{equation}
\label{eq:geuda-ex1}
\b_{B_1B_2B_3}=p_1\sigma_{B_1B_2B_3}+\sum_{m=2}^d p_m\proj{mmm},
\end{equation}
where $p_1>0$, $p_m\geq 0$, and
\begin{equation}
\label{eq:geuda-ex2}
\sigma_{B_1B_2B_3}=\frac{2}{3}\proj{\xi}+\frac{1}{3}\proj{111}
\end{equation}
via $\ket{\xi}=\frac{1}{2}\ket{010}+\frac{1}{2}\ket{110}+\frac{1}{\sqrt{2}}\ket{001}$.

Let $\a^1_{A_1A_2A_3}$ and $\a^2_{A_1A_2A_3}$ be two three-qubit $W$ type states, and $\b_{B_1B_2B_3}$ be a state given by Eq. \eqref{eq:geuda-ex1}. It follows from Proposition \ref{prop:geuda} that $\a^1_{A_1A_2A_3}\ox_K\a^2_{A_1A_2A_3}$, $\a^1_{A_1A_2A_3}\ox_{K_c}\a^2_{A_1A_2A_3}$ are pure genuinely entangled $2$-UDA states, and $\a^1_{A_1A_2A_3}\ox_K\b_{B_1B_2B_3}$, $\a^1_{A_1A_2A_3}\ox_{K_c}\b_{B_1B_2B_3}$ are mixed genuinely entangled $2$-UDA states. By repeating the composite process, we generate a class of genuinely entangled $2$-UDA states from a three-qubit $W$ type state and a tripartite state given by Eq. \eqref{eq:geuda-ex1}.

\end{example}

Due to the uniqueness, the genuine multipartite entanglement in the states given by Example \ref{ex:geuda-1} can be detected by measuring two-particle correlations only, which is experimentally realizable.

\section{Concluding remarks}
\label{sec:con}

The pure states that can be completely determined by their marginals are essential to the efficient QST. In this paper, we generalized the definition of $k$-UDA states to the context of arbitrary (no matter pure or mixed) states, motivated by the efficient QST of low-rank states. Similar to pure $k$-UDA states, we have shown that for mixed $k$-UDA states, the UDA property is also maintained under LU equivalence, and ``$k$-UDA'' also implies ``$(k+1)$-UDA''. Due to the demand for $k$-UDA states of large number of parties with a small number $k$, we considered the \emph{additivity} of $k$-UDA states via three different composite ways of tensor product. Two $k$-UDA states admit the additivity under each type of products, if the composite state for the corresponding product is still $k$-UDA. We have shown that under each type of products, the additivity holds when one of the two initial $k$-UDA states is pure. Specially, it implies that the additivity holds within the typical definition on pure states. We also proposed specific conditions to verify the additivity of two mixed $k$-UDA states. However, we conjectured there is generally no additivity for two mixed $k$-UDA states. Since one of the three composite ways, namely the $K_c$-product, is adopted to construct GME states, we derived an operational and effective method to construct $k$-UDA states with genuine entanglement in the systems of large number of parties, by uniting the construction of GME states and the additivity of $k$-UDA states. Using this method, we constructed a class of GME states which are uniquely determined by two-particle correlations only. The future work is to reveal more interesting properties of $k$-UDA states, and study the detection of genuine multipartite entanglement based on mixed UDA states.

\section*{acknowledgements}

Y. S. is supported by the Fundamental Research Funds for the Centeral Universities under Grant No. JUSRP123029. L. C. was supported by the NNSF of China (Grant No. 11871089), and the Fundamental Research Funds for the Central Universities (Grant No. ZG216S2005).


\bibliography{marginal}

\end{document}